%% file: arxiv.tex
\documentclass[11pt]{article}
\usepackage[utf8]{inputenc}
\usepackage[T1]{fontenc}
\usepackage{bm}
\usepackage{amsmath,amsfonts,amssymb,multirow}
\usepackage{amsthm}
\usepackage{multicol}
\usepackage{appendix}
\usepackage{cite}
\usepackage[margin=1in]{geometry}
\usepackage{graphics}
\usepackage[dvipsnames,usenames]{color}
\usepackage[colorlinks=true,urlcolor=Blue,citecolor=Green,linkcolor=BrickRed]{hyperref}
\usepackage[usenames,dvipsnames]{xcolor}
\usepackage{enumerate}
\usepackage{algorithm}
\usepackage[noend]{algorithmic}

\usepackage{paralist}
\usepackage{todonotes}

\usepackage{enumitem}

\usepackage{microtype}

\usepackage[utf8]{inputenc}
\usepackage{amsthm,thmtools,thm-restate}

\usepackage{lipsum}
\usepackage{todonotes}
\usepackage{thmtools,thm-restate}
\usepackage{tikz}
\usepackage[overload]{textcase}

\usepackage{algorithm}
\usepackage{algorithmic}
\usepackage{enumitem}

\newcounter{mycounter}

\newcommand{\Z}{\mathbb{Z}}

\newcommand{\rbrackets}[1]{\left( #1 \right)}
\newcommand{\sbrackets}[1]{\left[ #1 \right]}
\newcommand{\cbrackets}[1]{\left\{ #1 \right\}}

\newcommand{\Oh}[1]{O\!\rbrackets{#1}}
\newcommand{\Thetah}[1]{\Theta(#1)}
\newcommand{\Omegah}[1]{\Omega(#1)}

\newcommand{\floor}[1]{\lfloor #1 \rfloor}
\newcommand{\ceil}[1]{\lceil #1 \rceil}
\newcommand{\eps}{\varepsilon}

\newcommand{\domimg}{: \Z^+ \rightarrow \Z^+}
\newcommand{\rep}[1]{\expandafter\MakeUppercase\expandafter{#1}} 
\newcommand{\sumof}[1]{#1^{\leq}} \newcommand{\tildes}[1]{\tilde{\sumof{#1}}}
\newcommand{\shift}[2]{#1|_{#2}} 

\newcommand{\C}{C}
\newcommand{\knap}{k}

\newcommand{\Knap}{\rep{\knap}}

\newcommand{\apx}{\varepsilon}
\newcommand{\sparapx}{\delta}
\newcommand{\weight}{w}

\newcommand{\wset}{W}
\newcommand{\hminus}{\sqrt{n\eps}}
\newcommand{\hplus}{\sqrt{n/\eps}}
\newcommand{\di}{\delta_i}
\newcommand{\divalue}{\frac{\eps^{3/4}}{2c \cdot 2^{i/2} \cdot n^{1/4} }}

\newcommand{\countingknapsack}{{\sc \#Knapsack}}
\newcommand{\knapsack}{{\sc Knapsack}}

\newcommand{\Simpletime}{$O(n^3 \eps^{-1})$ }
\newcommand{\Simplespace}{$O(n^2 \eps^{-1})$}
\newcommand{\FPTAStime}{$O(n^{2.5}\varepsilon^{-1.5}\log(n \eps^{-1})\log (n \eps))$ }
\newcommand{\FPTASspace}{$O(n^{1.5}\varepsilon^{-1.5})$ }
\newcommand{\Integertime}{$O(n^{2.5}\varepsilon^{-1.5}\log(n\eps^{-1} \log U)\log (n \eps) \log^2 U)$ }
\newcommand{\Integerspace}{$O(n^{1.5}\varepsilon^{-1.5}\log U)$}

\DeclareMathOperator{\successor}{succ}

\newtheorem{lemma}{Lemma}
\newtheorem{theorem}{Theorem}

\newtheorem{definition}{Definition}

\makeatletter

\renewcommand{\theenumi}{\arabic{enumi}}

\renewcommand{\p@enumii}{\theenumi.}
\makeatother

\makeatletter

\begin{document}

\title{A Faster FPTAS for \#Knapsack}

\author{Pawe\l{} Gawrychowski\thanks{University of Wrocław, Poland.}
    \and
  Liran Markin\thanks{University of Haifa,  Israel. Partially supported by Israel Science Foundation grant 592/17.}
  \and
  Oren Weimann$^\dagger$
}

\date{}
\maketitle

\begin{abstract}
Given a set $\wset = \{w_1,\ldots, w_n\}$ of non-negative integer weights and an integer $\C$, the \countingknapsack\ problem asks to count the number of distinct
subsets of $\wset$ whose total weight is at most $\C$. In the more general integer version of the problem,  the subsets are multisets. That is, we are also given a set $ \{u_1,\ldots, u_n\}$ and we are allowed to take up to $u_i$ items of weight $w_i$.

We present a deterministic FPTAS for \countingknapsack\ running in \FPTAStime time.
The previous best deterministic algorithm [FOCS 2011] runs in $O(n^3 \varepsilon^{-1} \log(n\eps^{-1}))$ time (see also  [ESA 2014] for a  logarithmic factor improvement). The previous best randomized algorithm [STOC 2003] runs in $O(n^{2.5} \sqrt{\log (n\eps^{-1}) } + \varepsilon^{-2} n^2 )$ time. Therefore, in the natural setting of constant $\eps$, we close the gap between the $\tilde O(n^{2.5})$ randomized algorithm and the $\tilde O(n^3)$ deterministic algorithm.   

For the integer version with $U = \max_i \cbrackets{u_i}$, we present a deterministic FPTAS running in \Integertime time. The previous best deterministic algorithm [APPROX 2016] runs in $O(n^3\varepsilon^{-1}\log(n \eps^{-1} \log U) \log^2 U)$ time.
\end{abstract}

\section{Introduction}\label{sec:intro}

Given a set $\wset = \{w_1,\ldots, w_n\}$ of non-negative integer weights and an integer $\C$, the \countingknapsack\ problem asks to count the number of distinct
subsets of $\wset$ whose total weight is at most $\C$. This problem is the counting version of the well known \knapsack\ problem and is \#P-hard. 
While there are many, celebrated, randomized polynomial-time algorithms for approximately counting \#P-hard
problems, the \countingknapsack\ problem is one of the few examples where there is also a deterministic approximation algorithm (other notable examples are~\cite{halman2009fully,Wei06,BGK07}).

From a geometric view, the \countingknapsack\ problem is equivalent to finding the number of
vertices of the $n$-dimensional hypercube that lie on one side of a given $n$-dimensional hyperplane. 
The problem is also related to 
pseudorandom generators for
halfspaces (see e.g.~\cite{PRG1,PRG2,prg3}) as these imply deterministic (though not polynomial-time) approximation schemes for
\countingknapsack\ by enumerating
over all input seeds to the generator.  

\paragraph{Approximately counting knapsack solutions.}
The \countingknapsack\ problem can be solved with the following simple recursion: $S(i,j) = S(i-1,j)+S(i-1,j-w_i)$ where $S(i, j)$ is the number of subsets of $\{w_1,...,w_i\}$ whose weight sums to at most $j$. This recurrence immediately implies a pseudo-polynomial $O(n \C)$ time algorithm. More interestingly, this recurrence is the basis of all existing fully polynomial-time approximation schemes (FPTAS).   
That is, algorithms that for any $\varepsilon > 0$ estimate the
number of solutions to within relative error $(1 \pm \varepsilon)$ in time
polynomial in $n$ and in $1/\varepsilon$. 

Dyer et al.~\cite{dyer1993mildly} were the first to show how to approximate this recurrence with random sampling. They gave a randomized sub-exponential $2^{O(\sqrt{n}\log^{2.5} n))}\epsilon^{-2}$ time algorithm. 
Using a more complicated random sampling (with a rapidly mixing Markov chain), Morris and Sinclair~\cite{morris2004random} obtained the first FPRAS (fully-polynomial {\em randomized} approximation scheme) running in  $O(n^{4.5+\epsilon} + \epsilon^{-2} n^2)$  time. 
 Dyer~\cite{dyer2003approximate} further improved this to $O(n^{2.5} \sqrt{\log (n\eps^{-1}) } + \epsilon^{-2} n^2 )$ by using a surprisingly simple sampling procedure (combined with randomized rounding). This to date is the fastest known randomized solution.  As for deterministic solutions, the fastest solution to date is by Rizzi and Tomescu~\cite{rizzi2014faster} and runs in  $O(n^3 \epsilon^{-1} \log \epsilon^{-1} / \log n)$ time. It is a logarithmic factor improvement (obtained by discretizing the recursion $S(i,j)$ with floating-point arithmetics) over the previous fastest  $O(n^3 \epsilon^{-1} \log(n\eps^{-1}))$ time solutions of Gopalan et al.~\cite{gopalan2011fptas} (who used read-once branching programs inspired by related work on pseudorandom generators for halfspaces~\cite{PRG2}) and of {\v{S}}tefankovi{\v{c}} et al.~\cite{vstefankovivc2012deterministic} (who approximated a ``dual'' recursion $S^*(i,j)$ defined as the smallest capacity $c$ such that there exist at least $j$ subsets of $\{w_1, \ldots , w_i\}$ with weight $c$).

\paragraph{Approximately counting integer knapsack solutions.}
In the more general integer version of \countingknapsack, the subsets are multisets. That is, in addition to $\wset = \{w_1,\ldots, w_n\}$ we are also given a set $ \{u_1,\ldots, u_n\}$ and we are allowed to take up to $u_i$ items of weight $w_i$. 

The first (randomized) FPRAS for counting integer knapsack solution was given by Dyer~\cite{dyer2003approximate} who  presented a strongly polynomial $O(n^5 +n^4 \varepsilon^{-2})$ time algorithm. 
A (deterministic) FPTAS for this problem
was then given by Gopalan et al.~\cite{gopalan2011fptas} with a running time of $O(n^5\varepsilon^{-2} \log^2 U \log w)$ (see also~\cite{gopalan2010polynomial}) where $U = \max_i \cbrackets{u_i}$ and $w=\sum_i{w_iu_i}+C$.
The fastest solution to date is by Halman~\cite{halman2016deterministic} with a running time of $O(n^3\varepsilon^{-1}\log(n \eps^{-1} \log U) \log^2 U)$.

\paragraph{Our results.}

In this paper we present improved algorithms for both \countingknapsack\ and its integer version. Our algorithms improve the previous best algorithms by polynomial factors. 
For constant $\eps$, we close the gap between the $\tilde O(n^{2.5})$ randomized and the $\tilde O(n^3)$ deterministic running times.
More formally, with the standard assumption of constant time arithmetics on the input numbers, we prove the following two theorems: 

\begin{theorem}
    There is a FPTAS running in \FPTAStime time and \FPTASspace space for counting knapsack solutions.
\end{theorem}

\begin{theorem}
    There is a FPTAS running in \Integertime time and \Integerspace\ space for counting integer knapsack solutions. 
\end{theorem}

Our algorithm 
is the first algorithm to deviate from the standard recursion. In particular, on  large enough sets, instead of recursing on all but the last item, we recurse in the middle and use convolution to merge the two sub-solutions. This requires   
extending the recent technique of {\em K-approximation sets and functions} used by Halman~\cite{halman2016deterministic} and introduced in~\cite{halman2009fully}. 
Our extended technique (which we call {\em sum approximations}) is  simple to state and leads to a surprisingly simple solution to \countingknapsack\ with an improved running time.
In a nutshell, for any function $f \domimg$ (think of $f(x) = $ the number of subsets with total weight exactly $x$) let $\sumof{f}$ denote the function  $\sumof{f}(x) = \sum_{y \leq x} f(y)$ (hence $\sumof{f}(x) = $ the number of subsets with total weight at most $x$). Then, in order to approximate the function $\sumof{f}$ it is enough to find any function $F$ such that $\sumof{F}$ approximates $\sumof{f}$.

We examine the properties of such sum approximations $F$ in Section~\ref{sec:apx}, and introduce a number of useful computational primitives on sum approximations. With these primitives in hand, we give a simplified version of Halman's algorithm for \countingknapsack\ in Section~\ref{sec:simple}.
Then, in Section~\ref{sec:fptas} we present an improved divide and conquer algorithm based on convolutions of sum approximations.
Finally, in Section~\ref{sec:integer} we adapt our algorithm to the integer version, where every item has a corresponding multiplicity.
Instead of the {\em binding constraints} approach used by Halman~\cite{halman2016deterministic}, we show that it is enough to perform a single scan of a sum approximation using nothing more than a standard binary search tree.

\section{Approximation of a Function}\label{sec:apx}

 Consider the following two functions: 
$f(x) = $ the number of subsets with total weight exactly $x$, and $\sumof{f}(x) = $ the number of subsets with total weight at most $x$.  
More generally:    
 
\begin{definition}
	Given a function $f \domimg$ we define the function $\sumof{f}(x)$ as $$\sumof{f}(x) = \sum_{y \leq x} f(y).$$
\end{definition}

Our goal is to approximate $\sumof{f}(C)$ but we will actually approximate the entire function $\sumof{f}(x)$ for all $x$. We now describe what it means to approximate a function and present some properties of such approximations.

\begin{definition}[$(1+\apx)$-approximation of a function]

	Given a function $f \domimg$ and a parameter $\apx > 0$, a function $\rep{f} \domimg$  is a \emph{$(1+\apx)$-approximation of $f$} if for every $x$,
	\begin{equation*}
		f(x) \leq \rep{f}(x) \leq (1+\apx)f(x).
	\end{equation*}

\end{definition}

\begin{definition}[$(1+\apx)$-sum approximation of a function]\label{def:sumapx}

	Given a function $f \domimg$ and a parameter $\apx > 0$, a function $\rep{f} \domimg$ is a \emph{$(1+\apx)$-sum approximation of $f$} if $\sumof{\rep{f}}$ is a $(1+\apx)$-approximation of $\sumof{f}$.
\end{definition}

\noindent We next examine some useful properties of sum approximations. For a function $f \domimg$ define its {\em shift} by $w$ as follows:
\[\shift{f}{w}(x) = \begin{cases}
f(x-w), & x \geq w,\\
0 & x<w,
\end{cases}\]
and for two functions $f,g \domimg$ we define their {\em convolution} to be:
\[(f \ast g)(w) = \sum\limits_{x+y = w} f(x)g(y).\]

\begin{lemma}[operations on sum approximations] \label{lem:apxmath}
	Let $\rep{f}$ be a $(1+\apx)$-sum approximation of $f$ and $\rep{g}$ be a $(1+\apx)$-sum approximation of $g$, then the following properties hold:
	\begin{description}[leftmargin=!,labelwidth=\widthof{\bfseries Approximatio}]
		\item[Approximation:]\mbox{A $(1+\sparapx)$-sum approximation of $\rep{f}$ is a $(1+\sparapx)(1+\apx)$-sum approximation of $f$.}		\item[Summation:]$(\rep{f}+\rep{g})$ is a $(1+\apx)$-sum approximation of $(f+g)$.
		\item[Shifting:]$\shift{\rep{f}}{w}$ is a $(1+\apx)$-sum approximation of $\shift{f}{w}$ for any $w > 0$.		
		\item[Convolution:]$(\rep{f} \ast \rep{g})$ is a $(1 + \eps)^2$-sum approximation of $(f \ast g)$. 
	\end{description}
\end{lemma}

\begin{proof}
	\ 
	\begin{description}
		
		\item[Approximation:] Let $F'$ be a $(1+\sparapx)$-approximation of $\rep{f}$. For every $x$,  $\sumof{f}(x) \leq \sumof{\rep{f}}(x) \leq (1+\apx)\sumof{f}(x)$ and $\sumof{\rep{f}}(x) \leq \sumof{F'}(x) \leq (1+\sparapx)\sumof{\rep{f}}(x)$. We therefore have that $\sumof{f}(x) \leq \sumof{F'}(x) \leq (1+\sparapx)(1+\apx)\sumof{f}(x)$.

		\item[Summation:] For every $x$ we have that $\sumof{f}(x) \leq \sumof{\rep{f}}(x) \leq (1+\apx)\sumof{f}(x)$ and $\sumof{g}(x) \leq \sumof{\rep{g}}(x) \leq (1+\apx)\sumof{g}(x)$, adding these two equations we get $\sumof{(f+g)}(x) \leq \sumof{(\rep{f}+\rep{g})}(x) \leq (1+\apx)\sumof{(f+g)}(x)$.

		\item[Shifting:] For $x < w$, $\shift{f}{w}(x) = 0 = \shift{\rep{f}}{w}(x)$. For $x \geq w$ let $y = x - w$. Since $y \geq 0$ we have that $\sumof{f}(y) \leq \sumof{\rep{f}}(y) \leq (1+\apx)\sumof{f}(y)$ and therefore $\sumof{\shift{f}{w}}(x) \leq \sumof{\shift{\rep{f}}{w}}(x) \leq (1+\apx)\sumof{\shift{f}{w}}(x)$.
		
		\item[Convolution:]
We first prove that $\sumof{(\rep{f} \ast \rep{g})}(w) \geq \sumof{(f \ast g)}(w)$:
\[\begin{aligned}[t]
	 \sumof{(\rep{f} \ast \rep{g})}(w)  &=
	 \sum\limits_{x+y \leq w} \rep{f}(x)\rep{g}(y) = 
	 \sum\limits_{x \leq w} \sum\limits_{y \leq w-x} \rep{f}(x)\rep{g}(y) = 
	 \sum\limits_{x \leq w} \sbrackets{\rep{f}(x) \sum\limits_{y \leq w-x} \rep{g}(y)} \\
	&= \sum\limits_{x \leq w} \rep{f}(x) \sumof{\rep{g}}(w-x) \geq 
	 \sum\limits_{x \leq w} \rep{f}(x) \sumof{g}(w-x) = 
	 \sum\limits_{x \leq w} \rep{f}(x) \sum\limits_{y \leq w-x} g(y) \\
	& =\sum\limits_{x+y \leq w} \rep{f}(x) g(y) = 
	 \sum\limits_{y \leq w} \sum\limits_{x \leq w-y} \rep{f}(x) g(y) = 
	 \sum\limits_{y \leq w} \sbrackets{g(y) \sum\limits_{x \leq w-y} \rep{f}(x)} \\
	&= \sum\limits_{y \leq w} g(y) \sumof{\rep{f}}(w-y) \geq 
	 \sum\limits_{y \leq w} g(y) \sumof{f}(w-y) = 
	 \sum\limits_{y \leq w} g(y) \sum\limits_{x \leq w-y} f(x) \\
	&= \sum\limits_{x+y \leq w} f(x)g(y) = 
	 \sumof{(f \ast g)}(w).
\end{aligned}\]
\noindent Next we prove that $\sumof{(\rep{f} \ast \rep{g})}(w) \leq (1 + \apx)^2 \sumof{(f \ast g)}(w)$:
\[\begin{aligned}[t]
	 \sumof{(\rep{f} \ast \rep{g})}(w) &= 
	 \sum\limits_{x+y \leq w} \rep{f}(x)\rep{g}(y) = 
	 \sum\limits_{x \leq w} \sum\limits_{y \leq w-x} \rep{f}(x)\rep{g}(y) = 
	 \sum\limits_{x \leq w} \sbrackets{ \rep{f}(x) \sum\limits_{y \leq w-x} \rep{g}(y) } \\
	&= \sum\limits_{x \leq w} \rep{f}(x) \sumof{\rep{g}}(w-x) \leq 
	 \sum\limits_{x \leq w} \rep{f}(x) (1+\apx) \sumof{g}(w-x) =\\ &= 
	 (1+\apx) \sum\limits_{x \leq w} \rep{f}(x) \sum\limits_{y \leq w-x} g(y) = 
	 (1+\apx) \sum\limits_{x+y \leq w} \rep{f}(x) g(y) \\
	&= (1+\apx) \sum\limits_{\leq w} \sum\limits_{x \leq w-y} \rep{f}(x) g(y) = 
	 (1+\apx) \sum\limits_{y \leq w} \sbrackets{g(y) \sum\limits_{x \leq w-y} \rep{f}(x)} \\
	&= (1+\apx) \sum\limits_{y \leq w} g(y) \sumof{\rep{f}}(w-y) \leq 
 (1+\apx) \sum\limits_{y \leq w} g(y) (1+\apx) \sumof{f}(w-y) = \\
 &=	 (1+\apx)^2 \sum\limits_{y \leq w} g(y) \sum\limits_{x \leq w-y} f(x)  
	 = (1+\apx)^2 \sum\limits_{x+y \leq w} f(x)g(y) \\
	 &= (1+\apx)^2 \sumof{(f \ast g)}(w).&\qedhere 
\end{aligned}\]
	\end{description} 
\end{proof}

\noindent We next describe the way that we represent functions.
\begin{definition}[a function representation]
	Given a function $f \domimg$, the \emph{representation of $f$} is defined to be a list of all the pairs $(x, f(x))$ where $f(x) > 0$. The list is kept sorted by the $x$ value.
	The \emph{size} of $f$ (denoted by $|f|$) is the number of pairs in the representation of $f$. To simplify our presentation, we allow the representation to include multiple pairs with the same value of $x$. This can be easily fixed with a single scan over the  representation.
\end{definition}

In the following paragraphs we show how to efficiently implement the following operations on functions: \emph{sparsification}, \emph{summation}, \emph{shifting}, \emph{convolution}, and \emph{query}.
The output of each operation is a sum approximation.

\paragraph{Sparsification.}
Sparsification is the operation of constructing a $(1+\sparapx)$-sum approximation of $f$ (see Definition~\ref{def:sumapx}). The input is a function $f \domimg$ and a sparsification parameter $\sparapx > 0$, the output is a function $\rep{f} \domimg $ that is a $(1+\sparapx)$-sum approximation of $f$. The goal is to construct a function $\rep{f}$ that has a compact representation (i.e. a small number of points with non-zero values). To this end, we partition the values of $\sumof{f}$ into segments with elements belonging to $[r_{i},r_{i+1})$ (see Figure~\ref{fig:apx1}), where: 
\begin{align*}
	r_0 &= 0, 			 \\
	r_{i+1} & = \max\{ r_i + 1, \floor{ (1 + \sparapx)r_i } \} .
\end{align*}

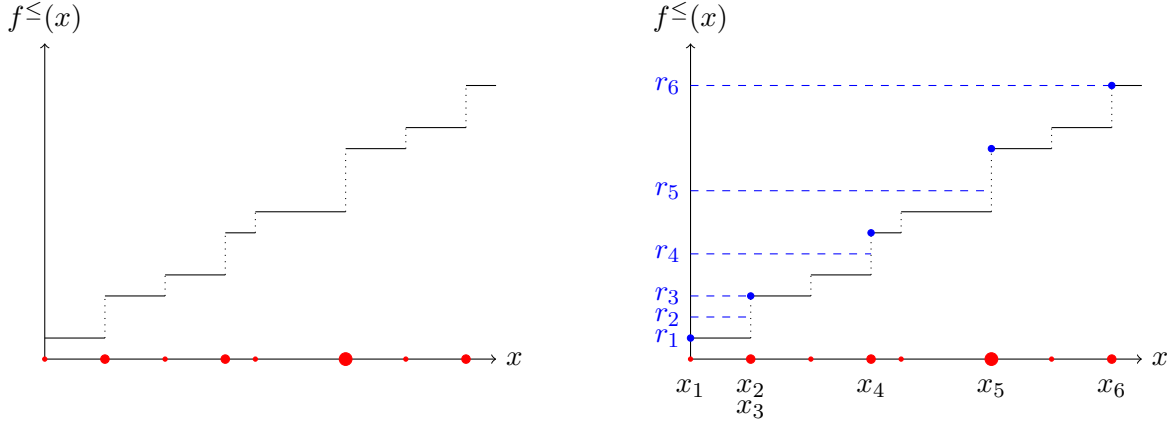
\begin{figure}[t]
\centering
\begin{minipage}{.48\textwidth}
  \centering
  \begin{tikzpicture}[scale=0.4, x=1cm, y=0.7cm]
  	\input{graphics/apx1} 
  	  \end{tikzpicture}
\end{minipage}\hfill
\begin{minipage}{.48\textwidth}
  \centering
  \begin{tikzpicture}[scale=0.4, x=1cm, y=0.7cm]
  	\input{graphics/apx2}  
  \end{tikzpicture} 
\end{minipage}
\caption{On the left,  $\sumof{f}(x)$ compared to $f(x)$. The red point at position $x$ is wider as $f(x)$ is larger.  On the right, the blue points are the first entries that have value of at least $r_i$ \label{fig:apx1}}
\end{figure}

\noindent  We call $x_i = \min_{x} \cbrackets{ \sumof{f}(x) \geq r_i } $ the $i$-th {\em breakpoint}.
For any $x$, let $\successor(x)$ be the strict successor of $x$ among $\{x_i\}$, i.e. $\successor(x) = \min_{i} \cbrackets{ x_i > x }$.
We define the function $\tildes{f}$ (see Figure~\ref{fig:apx3}) as: $$\tildes{f}(x) = \sumof{f}(\successor(x)-1),$$
where $\tildes{f}(x) = \lim_{x\to\inf} \sumof{f}(x)$ if $\successor(x)=\infty$.

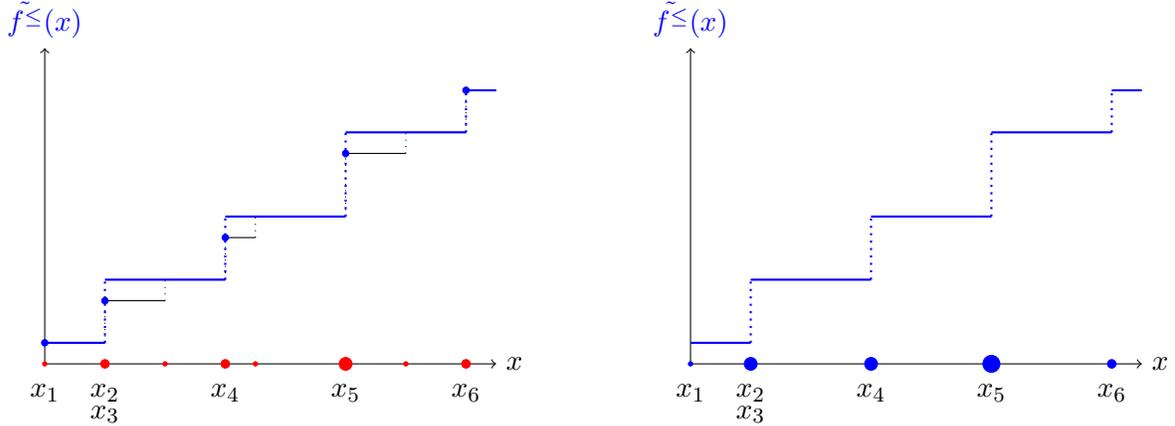
\begin{figure}[h!]
\centering
\begin{minipage}{.48\textwidth}
  \centering
  \begin{tikzpicture}[scale=0.4, x=1cm, y=0.7cm]
  	\input{graphics/apx3} 
  \end{tikzpicture}
\end{minipage}\hfill
\begin{minipage}{.48\textwidth}
  \centering
  \begin{tikzpicture}[scale=0.4, x=1cm, y=0.7cm]
  	\input{graphics/apx4} 
  \end{tikzpicture}
\end{minipage}
 \caption{On the left, $\tildes{f}$ (in blue) is defined from $\sumof{f}$ and has the same value in any segment $[x_i,x_{i+1})$. On the right, the construction of $\rep{f}(x)$. The blue points are only at positions $x_i$ and are wider as $\rep{f}(x_i)$ is larger.\label{fig:apx3}}
\end{figure}

\begin{lemma}\label{lemma:tildes:f}
	$\tildes{f}$ is a $(1 + \sparapx)$-approximation of $\sumof{f}$.
\end{lemma}

\begin{proof}
First observe that $\sumof{f}(x) \leq \tildes{f}(x)$ (since  $\successor(x) > x$ and $\sumof{f}$ is monotone). Consider any $x$ and let $i$ be the unique index such that $r_i \leq \sumof{f}(x) < r_{i+1}$. If $\successor(x)=\infty$ then $\tildes{f}(x) = \lim_{x\to\inf} \sumof{f}(x) < r_{i+1}$.
Otherwise, $x_{i+1}>x$ and $\tildes{f}(x) = \sumof{f}(x_{i+1}-1) < r_{i+1}$. We need to consider two cases:
	If $r_{i+1} \leq (1 + \sparapx)r_i$, then $\tildes{f}(x) < r_{i+1} \leq (1 + \sparapx)\sumof{f}(x)$.
If $r_{i+1} = r_i + 1$ and because the values of $\tildes{f}(x)$ are integer, $\tildes{f}(x) \leq r_{i+1}-1 = r_{i}$. So in both cases
$\tildes{f}(x) \leq  (1 + \sparapx)\sumof{f}(x)$.
\end{proof}

We can now define the function $\rep{f}$ (the $(1+\sparapx)$-sum approximation of $f$).
Observe that by Lemma~\ref{lemma:tildes:f}, every  $\rep{f}$ such that $\sumof{\rep{f}} = \tildes{f}$ is a  $(1+\sparapx)$-sum approximation.
We define $\rep{f}$ as the discrete derivative of $\tildes{f}$. That is,
\begin{equation*}
	\rep{f}(x) = \begin{cases}
	\tildes{f}(x) - \tildes{f}(x-1)     &x > 0, \\
	\tildes{f}(x)						&x = 0.
\end{cases}
\end{equation*}
It is easy to see that $\sumof{\rep{f}} = \tildes{f}$. It is also easy to construct the representation of $\rep{f}$ in linear time with a single scan over the representation of $f$ (see Algorithm~\ref{alg:sparsification}).
Let $M$ be the maximum value of $\sumof{f}$. 
Notice that $\tildes{f}$ can have at most $|\{x_i\}| = |\{r_i\}| = \log_{1+\sparapx} M$ different values. This means that 
 $|\rep{f}|$ is at most $\log_{1+\sparapx} M$.
The total running time is therefore  $O(|f| + \log_{1+\sparapx} M)$. 

\begin{algorithm}[h]
\renewcommand{\algorithmicrequire}{\textbf{Input:}}
\renewcommand{\algorithmicensure}{\textbf{Output:}}
\caption{\textsc{Sparsify}($f,\sparapx$)}
\label{alg:sparsification}
\begin{algorithmic}[1]
\REQUIRE a function $f$ represented by a sorted list of all pairs $(x,f(x))$ where $f(x) > 0$ and a sparsification parameter $\sparapx > 0$.
\ENSURE a function $\rep{f}$ that is a $(1+\sparapx)$-sum approximation of $f$ and is represented by a sorted list of at most $\log_{1+\sparapx} M$ pairs (where $M$ is the maximum value of $\sumof{f}$).
\STATE initialize $r = accum = prevaccum = prevx=0$
\FOR {every pair $(x,f(x))$ in sorted order} 
 \STATE $r \gets \max\{ r + 1, \floor{ (1 + \sparapx)r } \}$ 
 \WHILE{$accum < r$ }
 \STATE $accum \gets accum + f(x)$
 \STATE get the next pair $(x,f(x))$ in the list
 \ENDWHILE
 \STATE add the pair $(prevx,accum-f(x)-prevaccum)$ to $F$
  \STATE $prevaccum \gets accum-f(x)$
  \STATE $prevx \gets x$
\ENDFOR
\end{algorithmic}
\end{algorithm}

\paragraph{Summation.}
Given two $(1+\apx)$-sum approximations $\rep{f}$ and $\rep{g}$ of functions $f$ and $g$ respectively, we wish to construct the function $\rep{f} + \rep{g}$ (that is a $(1+\apx)$-sum approximation of $f+g$ by Lemma~\ref{lem:apxmath}).
We construct $\rep{f}+\rep{g}$ naively by setting $(\rep{f}+\rep{g})(x) = \rep{f}(x) + \rep{g}(x)$. The sorted list of $\rep{f} + \rep{g}$ can be obtained in linear time given two sorted lists of $\rep{f}$ and of $\rep{g}$.
The total space and time is therefore $O(|\rep{f}|+|\rep{g}|)$.

\paragraph{Shifting.}
Given a $(1+\apx)$-sum approximation $\rep{f}$ of $f$, the function shifted by $w$, $\shift{\rep{f}}{w}$ is a $(1+\apx)$-sum approximation of $\shift{f}{w}$ by Lemma~\ref{lem:apxmath}.
In order to create $\shift{\rep{f}}{w}$, we take every pair in the representation of $\rep{f}$, namely $(x,y = \rep{f}(x))$ and change it to $(x+w,y)$. $\shift{\rep{f}}{w}(x)$ will be the sum of all the pairs where the first coordinate is $x$.
The total space and time is  $O(|\rep{f}|)$.

\paragraph{Convolution.}
The convolution $\rep{f} \ast \rep{g}$ contains all the combinations of taking some $x$ value from $\rep{f}$ and some $y$ value from $\rep{g}$. For every pair $x,y$ such that $\rep{f}(x) \neq 0$ and $\rep{g}(x) \neq 0$ we add the value $\rep{f}(x) \rep{g}(y)$ to the value of $\rep{f} \ast \rep{g}$ at point $x+y$.

We sort these pairs in order to get the representation of $\rep{f} \ast \rep{g}$.
For a certain $y$ we have all the points $(x + y, \rep{f}(x))$ sorted already, those are $|\rep{g}|$ sorted sequences that we have to merge. 
The total space of the output $\rep{f} \ast \rep{g}$ is at most the number of such pairs $x,y$, that is $|\rep{f}| \cdot |\rep{g}|$.
But it is possible to obtain a stream of the sorted pairs with their value using less space, by using a heap to merge the lists. Assuming without loss of generality that $|\rep{g}| \leq |\rep{f}|$, each list is a value $y$ and a pointer to a point in $\rep{f}$, and the heap extracts the minimum value of the sum of $y$ and the value in the pointer. 
The total time to create $\rep{f} \ast \rep{g}$ is therefore $O(|\rep{f}| \cdot |\rep{g}| \cdot \log (\min\{ |\rep{f}|, |\rep{g}| \} ) )$ and the space $O(|\rep{f}| + |\rep{g}|)$.

\paragraph{Query.}
Given a $(1+\apx)$-sum approximation $\rep{f}$ of $f$ and a point $x$, we can query the value $\sumof{\rep{f}}(x)$ that satisfies $\sumof{f}(x) \leq \sumof{\rep{f}}(x) \leq (1 + \apx)\sumof{f}(x)$ in time  $O(|\rep{f}|)$.
This is because computing the function $\sumof{\rep{f}}(x) = \sum_{y \leq x} \rep{f}(y)$ takes  $O(|\rep{f}|)$ time by considering every $y$. Moreover, if we store the representation of $\rep{f}$ in a balanced binary search tree $T$ then a query can be done in $O(\log |\rep{f}|)$ time with a prefix sum query on $T$.

\section{The Algorithm of Halman~\cite{halman2016deterministic} (Simplified)}\label{sec:simple}

In this section we present a simplified version of the algorithm of Halman~\cite{halman2016deterministic} for \countingknapsack\ using sum approximations. 
The running time of this simple deterministic algorithm is \Simpletime and the space is \Simplespace.

For a set of weights $S$, let $\knap_{S}(x)$ denote the number of subsets of $S$ with total weight exactly $x$. The output of the algorithm is the function $\Knap_{\wset}$ that is a $(1+\eps)$-sum approximation of $\knap_{\wset}$.
The desired answer, $\sumof{\Knap}_{\wset}(\C)$, can then be easily  obtained using the query operation.

Recall that $\shift{\knap_S}{w}(x) = \knap_S(x-w)$ if  $x \geq w$ and 0 otherwise. The algorithm is based on the following observation:
\begin{lemma}
	Let $S$ be a set of integer weights and $w$ be an additional integer weight, then:
	\begin{equation*}\label{lem:simple}
		\knap_{S \cup \{w\}} = \knap_S + \shift{\knap_S}{w}
	\end{equation*}
\end{lemma}

\begin{proof}
Any subset of $S \cup \{w\}$ with weight $x$ either includes $w$ (the number of such solutions is $\knap_S(x-w)$) or does not include $w$ (the number of such solutions is $\knap_S(x)$).
Since these options are disjoint, we have that $\knap_{S \cup \{w\}}(x) = \knap_S(x) + \shift{\knap_S}{w}(x)$.
\end{proof}

\paragraph{The algorithm.}
The algorithm uses the above lemma to construct the set $S$ by inserting one element at a time (until $S=W$), keeping $\Knap_S$ updated. The algorithm starts by setting $\Knap_{\emptyset}(0) = 1$ and $\Knap_{\emptyset}(x) = 0$ for any $x\neq 0$.  
In the $i$-th step,
we compute the function $\Knap_{\{w_1, \ldots, w_{i}\}}$ from $\Knap_{\{w_1, \ldots, w_{i-1}\}}$.
Computing $\Knap_{\{w_1, \ldots, w_{i}\}}$ can be done with one shifting and one summation operation: 
$\Knap_{\{w_1, \ldots, w_{i}\}} = \Knap_{\{w_1, \ldots, w_{i-1}\}} + \shift{\Knap_{\{w_1, \ldots, w_{i-1}\}}}{w_i}$.
Notice that the size $|\Knap_{\{w_1, \ldots, w_{i}\}}| = 2|\Knap_{\{w_1, \ldots, w_{i-1}\}}|$ doubles from the summation operation. To overcome this blowup, at the end of each step of the algorithm, we sparsify with  parameter $\sparapx = (1 + \eps)^{\frac{1}{n}} - 1$.

\paragraph{Correctness.}
From Lemmas~\ref{lem:apxmath} and~\ref{lem:simple}, it follows that if $\Knap_S$ is a $(1+\alpha)$-sum approximation of $\knap_S$, then $\Knap_S + \shift{\Knap_S}{w}$ is a $(1+\alpha)$-sum approximation of $\knap_{S \cup \{w\}}$. 
Furthermore, the approximation factor of $ \Knap_{\{w_1, \ldots, w_i\}} $ after the sparsification is the approximation factor of $\Knap_{\{w_1, \ldots, w_{i-1}\}}$ multiplied by $(1 + \sparapx)$. We get that $\Knap_{\wset}$ is a $((1 + \eps)^{\frac{1}{n}})^n$-sum approximation of $\knap_{\wset}$, as required.

\paragraph{Time complexity.}
The size of $\Knap_{\{w_1, \ldots, w_i\}}$ after the sparsification is bounded by $ \log_{1+\sparapx} 2^i$.
The time complexity is therefore:
\begin{equation*}
	\sum\limits_{i = 1}^{n} \Oh{ \log_{1+\sparapx} 2^i} =
	\Oh{ \sum\limits_{i=1}^{n} \frac{i}{\log (1 + \sparapx) } } =
	\Oh{\frac{1}{\frac{1}{n}\log(1 + \eps)} \sum\limits_{i=1}^{n} i } =
	\Oh{n^3 \eps^{-1}}.
\end{equation*}
\paragraph{Space complexity.}
The space is  $O(|\Knap_{\wset}|) = O(\log_{1+\sparapx} 2^n) = 
\Oh{\frac{n}{\log (1+\sparapx)}} = \Oh{\frac{n^2}{\log (1+\eps)}} = O(n^2 \eps^{-1})$,
where we have used that $\ln(1+\eps)\geq \eps/2$ for $\eps\in (0,1)$.

\section{The Algorithm for Counting Knapsack Solutions}\label{sec:fptas}
In this section we present a deterministic \FPTAStime time \FPTASspace space
algorithm for counting knapsack solutions. 
The algorithm is based upon a similar observation to the one in Lemma~\ref{lem:simple}:
\begin{lemma}\label{lem:fptas}
	Let $S$ and $T$ be two sets of integer weights, then:
	\begin{equation*}
		\knap_{S \cup T} = \knap_S \ast \knap_T
	\end{equation*}
\end{lemma}

\begin{proof}
A subset of $S \cup T$ of weight $\weight$ must be obtained by taking a subset of weight $\weight_S$ from $S$ and a subset of weight $\weight_T$ from $T$ where $\weight_S + \weight_T = \weight$.
Thus, $\knap_S(\weight_S)$ subsets of $S$ of weight $\weight_S$ and $\knap_T(\weight_T)$ subsets of $T$ weight $\weight_T$
generate $\knap_S(\weight_S)\knap_T(\weight_T)$ subsets of $S \cup T$ of weight $\weight_S + \weight_T$.
Overall, we get that $\knap_{S \cup T}(\weight) = \sum\limits_{\weight_S + \weight_T = \weight} \knap_S(\weight_S)\knap_T(\weight_T)=(\knap_S \ast \knap_T)(\weight).$
\end{proof}

\paragraph{The algorithm.}
As in Section~\ref{sec:simple}, our algorithm computes a $(1 + \varepsilon)$-sum approximation $\Knap_W$ of $\knap_W$ recursively. This time however, we do two things differently: (1) The value of the approximation factor is different for each recursion depth. In a recursive call of depth $i$, we are given a set $S$ and we compute a $(1 + \eps_i)$-sum approximation $\Knap_S$ of $\knap_S$ for some $\eps_i$ to be chosen later. (2) Given a set $S$ we recurse differently depending on the size of $S$:

\begin{enumerate}
	\item If $|S| > \hplus$, then we partition the set $S$ into two sets $A$ and $B$ each of size $|S|/2 = n/2^i$ and make two recursive calls: One computes a $(1 + \eps_{i+1})$-sum approximation $\Knap_A$ of $\knap_A$ and the other computes a $(1 + \eps_{i+1})$-sum approximation $\Knap_B$ of $\knap_B$.  We then find $\Knap_S$ by computing the convolution $\Knap_A \ast \Knap_B$ and sparsifying with parameter 
\begin{equation*}
	\di = \divalue,
\end{equation*}
where $\varepsilon$ is the original approximation parameter and  $c = \frac{\sqrt 2}{\sqrt 2 - 1}$.

\item 
If $|S| \leq \hplus$, then we apply the algorithm from Section~\ref{sec:simple} on the set $S$ with parameter $\delta_{\log(\hminus)} = \Omegah{\sqrt{\eps / n}}$. Observe that in such a case the recursion depth is at least $\log (\hminus)$.
\end{enumerate}

\paragraph{Correctness.}
From Lemmas~\ref{lem:apxmath} and~\ref{lem:fptas}, it follows that if $\Knap_A$ is a $(1 + \eps_{i+1})$-sum approximation of $\knap_A$ and $\Knap_B$ is a $(1 + \eps_{i+1})$-sum approximation of $\knap_B$, then $\Knap_S$ is a $(1 + \eps_{i+1})^2 (1+ \di)$-sum approximation of $\knap_S$.
This means that 
$\eps_i$ satisfies the following relation:
\begin{equation*}
	(1 + \eps_i) = (1 + \eps_{i+1})^2 (1 + \di)
\end{equation*}
From the above equation and since on the bottom of the recursion with $\delta_i =  \delta_{\log (\hminus)}$, the final approximation factor of $\Knap_{\wset}$ is:
\begin{equation*}
	(1 + \eps_0) = (1 + \delta_{\log (\hminus)})^{\hminus} \prod\limits_{i = 0}^{\log(\hminus) - 1} (1 + \di)^{2^i} = \prod\limits_{i = 0}^{\log (\hminus)} (1 + \di)^{2^i}
\end{equation*}
We need to prove that the above product is not larger than $(1 + \eps)$. 
Since $	\sum\limits_{i = 0}^{x} 2^{i/2} < \frac{\sqrt 2}{\sqrt 2 - 1} \cdot 2^{x/2} = c \cdot 2^{x/2}
$ and $
	(1 + \di)^{2^i} = (1+\di)^{1/\di \cdot 2^{i}\di} \leq e^{2^{i}\di}= e^{\frac{\eps^{3/4}}{2c} n^{-1/4} 2^{i/2}}$ we obtain:
\begin{equation*}
	\prod\limits_{i = 0}^{\log (\hminus)} (1 + \di)^{2^i} \le e^{\frac{\eps^{3/4}}{2c} n^{-1/2} \sum_{i = 0}^{\log (\hminus)} 2^{i/2}} < e^{\frac{\eps^{3/4}}{2} n^{-1/4} 2^{\log (\hminus)/2}} = e^{\frac{\eps}{2}} \leq (1 + \eps),
\end{equation*}
where the last inequality follows from $\ln(1+\eps)\geq \eps/2$ for $\eps\in (0,1)$. Moreover, since the recursion changes at depth $\log (\hminus)$ we further need to assume that $\varepsilon \ge 1/n$. These assumptions are without loss of generality since for $\varepsilon > 1 $ we could simply use $\varepsilon = 1$ and for $\varepsilon < 1/n$ the previous algorithms are faster.

\paragraph{Time complexity.}
We analyze the time complexity of every recursion depth $i$.
For depth $i = \log (\hminus)$, we apply the simple algorithm from Section~\ref{sec:simple} $\floor{\hminus}$ times on sets of size $\Thetah{\hplus}$ with $ \delta_{\log (\hminus)} = \Omegah{\sqrt{\eps / n}}$, the running time is therefore $O({n^{2.5}}{\eps^{-1.5}})$. 

For depth $i < \log(\hminus)$, we apply $2^i$ convolutions and sparsifications. The time of the convolutions is dominant.
The total running time is therefore:
\begin{align*}
\sum\limits_{i = 0}^{(\log \hminus) - 1} 2^i \rbrackets{\frac{n}{2^{i+1} \cdot \delta_{i+1}}}^2 \log \rbrackets{\frac{n}{2^{i+1} \cdot \delta_{i+1}}} &\leq 
	\frac{1}{2} \sum\limits_{i = 1}^{\log \hminus} 2^i \rbrackets{\frac{n}{2^i \cdot \di}}^2 \log (n\eps^{-1}) \\ 
	&= \Oh{\sum\limits_{i = 1}^{\log \hminus} \frac{n^{2.5}}{\eps^{1.5}} \log (n\eps^{-1}) } \\
	&= O(n^{2.5}\varepsilon^{-1.5}\log(n\eps^{-1})\log (n \eps)).
\end{align*}

\paragraph{Space complexity.}
Since each recursive call makes at most two recursive calls, we do not need to keep more than two representation of sum approximations
on every level of the recursion.
We have seen in Section~\ref{sec:apx} that it is possible to construct sparsification of convolution in linear space.
The space complexity of all recursive calls of depth $i < \log(\hminus)$ is therefore:
\begin{equation*}
	\sum\limits_{i = 0}^{\log (\hminus)} \rbrackets{ 2 \cdot \frac{n}{2^i \cdot \di} } = \Oh{\sum\limits_{i = 0}^{\log (\hminus)} \frac{n^{1.25}}{2^{i/2} \cdot \eps^{0.75}} } = \Oh{{n^{1.25}}{\eps^{-0.75}}}
\end{equation*}
One call to the algorithm from Section~\ref{sec:simple} uses $
\Oh{\frac{\rbrackets{\hplus}^2}{\delta_{\log (\hminus)}}} =
\Oh{\frac{\rbrackets{\hplus}^2}{\sqrt{\eps / n}}} = O({n^{1.5}}{\eps^{-1.5}})$ space, and therefore the total space complexity is \FPTASspace.

\section{The Algorithm for Counting Integer Knapsack Solutions}\label{sec:integer}
In this section we show how to generalize the  algorithms of Sections~\ref{sec:simple} and~\ref{sec:fptas} to the integer version of counting knapsack solutions.

\subsection{Generalizing the algorithm of Section~\ref{sec:simple}}
In Section~\ref{sec:simple} we showed how to insert into a set $S$ a single item with weight $w$. 
We now need to show how to insert a single item with weight $w$ and multiplicity $u$.
The proof of the following lemma is similar to that of Lemma~\ref{lem:simple}.
\begin{lemma}
	Let $S$ be a set of integer pairs representing weights and multiplicity of items, and let $(w,u)$ be the weight and multiplicity of an additional item, then:
	\begin{equation*}
		 \knap_{S \cup \cbrackets{(w,u)}} = \knap_S + \shift{\knap_S}{w} + \shift{\knap_S}{2w} + \ldots + \shift{\knap_S}{u \cdot w}
	\end{equation*}
\end{lemma}

We will describe a new operation on sum approximations that creates the sparsification of $\rep{g} =(\Knap_S + \shift{\Knap_S}{w} + \shift{\Knap_S}{2w} + \ldots + \shift{\Knap_S}{u \cdot w})$ without actually computing $\rep{g}$, i.e. without actually computing all the points with non-zero value.

\paragraph{Events.}

Observe that a point $(x,y)$ with non-zero value $y = \Knap_S(x)$ implies $u+1$ points in the above sum: $(x, y), (x + w, y), (x + 2w, y), \ldots, (x + u \cdot w, y)$.
We call the first point $(x,y)$ a {\em start event} (with  position $x$ and value $y$) and the last point $(x + u \cdot w,y)$  an {\em end event} (with position $x + u \cdot w$ and value $y$).
Overall, for every $x$ with non-zero value $y = \Knap_S(x)$ there are two events, a total of $2|\Knap_S|$ events. It is possible to sort in linear time the sequence of events by their positions $\{x\}_i \cup \{x + u \cdot w\}_i$  because $\Knap_S$ is given sorted. We call the sorted list of events the {\em event list}. 

Similarly to Algorithm~\ref{alg:sparsification}, we could construct the sparsification of $\rep{g} = (\Knap_S + \shift{\Knap_S}{w} + \shift{\Knap_S}{2w} + \ldots + \shift{\Knap_S}{u \cdot w})$ by scanning all points in $\rep{g}$. This however would be too costly. Instead, we next present a new operation that constructs the sparsification of $\rep{g}$ while only scanning the $O(|\Knap_S|)$ events in the event list. 

\paragraph{\textsc{InsertAndSparsify}.}
 While scanning the event list, when we see a start event $(x, y)$ then we say that this event is an {\em active event} and that all the points $(x, y), (x + w, y), (x + 2w, y), \ldots, (x + u \cdot w, y)$ are {\em active points}. These points will become inactive when we will see the end event $(x + u \cdot w, y)$. As in Algorithm~\ref{alg:sparsification}, we would like to accumulate the values of all points seen so far. When the accumulator is larger than $r$, we introduce a new breakpoint (i.e., output a new point to the sparsification of $\rep{g}$) and set the new $r$ to be $\max\{ r + 1, \floor{ (1 + \sparapx)r } \}$. During the scan, apart from the accumulator, we also maintain a value $Y=$ the sum of values of all currently active events.
 
 When we scan a start event $(x,y)$, we add the value $y$ to both the accumulator and to $Y$.  This is not enough. We also need to add to the accumulator the values of all active points whose position is $x$. These are precisely the points whose start events were at positions $x-w, x-2w,\ldots, x-u\cdot w$. Notice that all these points have the same start position modulo $w$. For this reason, we maintain a balanced binary search tree $T$ that stores all active events keyed by their start position modulo $w$. Each node $v$ in $T$ with key $r\in \{0,\ldots,w-1\}$ stores a field $Y_v=$ the sum of values of all active points whose start position is $x$ such that $x=r \pmod w$. When we see a start event $(x,y)$, we search in $T$ for the node $v$ with key $x \bmod w$ (or create one if no such node exists), we increase the node's $Y_v$ field by $y$, and increase the accumulator by $y$. If $(x,y)$ is an end event we subtract $y$ from $Y_v$ and from $Y$.
 
After processing event $(x,y)$ as explained above, we want to process the next event $(x',y')$ in the event list. Before doing so, we need to: (1) increase the accumulator by the total value of all points in the segment $[x, x')$, and (2) if the updated accumulator is larger than $r$, find and output the (possibly many) new breakpoints whose positions are between $x$ and $x'$.\\\\ 
(1) We partition the segment $[x,x')$ into three segments: $[x, k_1 w), [k_1 w, k_2 w), [k_2 w, x')$ such that the lengths of the first and last segments are smaller than $w$, and the length of the middle segment is a multiple of $w$. Notice that every segment of length $w$ contains an active point from every active event exactly once. Therefore, to obtain the total value of points in the segment $[x, x')$ we query $T$ for the total value in segment $[x, k_1 w)$ (with a {\em suffix sum} query on the $Y_v$ values) and in segment $[k_2 w, x')$ (with a {\em prefix sum} query on the $Y_v$ values), and add to it $Y \cdot (k_2 - k_1)$ (the total value in segment $[k_1 w, k_2 w)$).\\\\
(2)  After increasing the accumulator by the above total value, if the accumulator becomes larger than $r$, then we will find all the new breakpoints in $[x, x')$ in $O(\log |\Knap_S|)$ time per breakpoint. Suppose the accumulated value at $x$ was $prevaccum$. To find the first breakpoint in the segment $[x, k_1 w)$ we query $T$ for the the first node after $x \bmod w$ such that the sum of $Y_v$ values between $x \bmod w$ and that node is at least $r - prevaccum$ (we call this a {\em succeeding sum} query on $T$,  and we symmetrically define a {\em preceding sum} query in which we seek the first node before $x \bmod w$ rather than after). We then output this breakpoint, set $prevaccum$ to be the accumulated value in this breakpoint, set $r = \max\{ r + 1, \floor{ (1 + \sparapx)r } \}$, and continue in the same way to find the next breakpoint in the segment $[x, k_1 w)$. Next, we find the breakpoints in the segment $[k_1 w, k_2 w)$. Since this segment is composed of $(k_2 - k_1)$ subsegments of length $w$, and since each of these subsegments contributes exactly $Y$ to the accumulator, it is easy (in $O(1)$ time) to find which subsegment contains the next breakpoint. On this subsegment we proceed similarly as on $[x, k_1 w)$. Finally, we need to find the breakpoints in segment $[k_2 w, x')$, again similarly as in $[x, k_1 w)$.
See Algorithm~\ref{alg:insertandsparsify}.

\paragraph{Time and space complexity.}
Each operation on the binary search tree $T$ takes $O(\log |\Knap_S|)$ time thus the total time for \textsc{InsertAndSparsify} is $\Oh{(|\Knap_{S \cup \{(w, u)\}}| + |\Knap_S|) \cdot \log |\Knap_S|}$. If $S=\{w_1, \ldots, w_i\}$ then $|\Knap_S| = \log_{1 + \sparapx} U^{i}$ and $|\Knap_{S \cup \{(w, u)\}}| = \Oh{\log_{1 + \sparapx} U^{i}}$.  The total time complexity of the algorithm is therefore:

\begin{align*}
\sum\limits_{i = 1}^{n} \Oh{ \log_{1+\sparapx} U^i \cdot \log (\log_{1+\sparapx} U^i}) &=
	\Oh{ \sum\limits_{i=1}^{n} \frac{i\log U}{\log (1 + \sparapx)} \cdot \log\left(\frac{i\log U}{\log (1 + \sparapx)}\right) } \\
		&= \Oh{\frac{\log U}{\frac{1}{n}\log(1 + \eps)} \cdot \log\left(\frac{n\log U}{\frac{1}{n}\log(1 + \eps)}\right)  \sum\limits_{i=1}^{n} i } \\
		& =\Oh{n^3 \eps^{-1} \log U \log \rbrackets{n\eps^{-1} \log U}}.
\end{align*}

\noindent The space complexity is  $\Oh{ |\Knap_{\wset}| } = \Oh{ \log_{1+\sparapx} U^n} = 
\Oh{\frac{n\log U}{\log (1+\sparapx)}}  = O(n^2 \eps^{-1} \log U)$.

\begin{algorithm}[h!]
\renewcommand{\algorithmicrequire}{\textbf{Input:}}
\renewcommand{\algorithmicensure}{\textbf{Output:}}
\caption{\textsc{InsertAndSparsify}($\Knap_S, w, u, \sparapx$)}
\label{alg:insertandsparsify}
\begin{algorithmic}[1]
\REQUIRE a sum approximation $\Knap_S$ of $\knap_S$, a new item with weight of $w$ and multiplicity $u$, and a sparsification parameter $\sparapx$.
\ENSURE a function $\rep{g}$ that is a $(1+\sparapx)$-sum approximation of $\knap_{S \cup (w, u)}$.
\STATE initialize $T$ as an empty binary search tree of pairs $(x, y)$ indexed by $(x \bmod w)$
\STATE initialize $Y = x = accum = prevaccum = 0$, and $r = 1$
\FOR {every event $(x', y')$ in sorted order} 
	\STATE $k_1 \gets \ceil{\frac{x}{w}}$, $k_2 \gets \floor{\frac{x'}{w}}$
	\WHILE{$accum + T.\textsf{suffixSum}(x) + Y \cdot (k_2 - k_1) + T.\textsf{prefixSum}(x' - 1) \geq r$ }
		\IF{$accum + T.\textsf{suffixSum}(x) \geq r$}
			\STATE $bp \gets T.\textsf{succeedingSum}(x, r - accum)$
			\STATE $accum \gets accum + T.\textsf{suffixSum}(x) - T.\textsf{suffixSum}(bp)$
		\ELSE
			\STATE $k_3 \gets k_1 + \floor{(r - accum - T.\textsf{suffixSum}(x))/Y}$
			\STATE $bp \gets T.\textsf{precedingSum}(x', r - (accum + T.\textsf{suffixSum}(x) + Y \cdot (k_3 - k_1)))$ 
			\STATE $accum \gets accum + T.\textsf{suffixSum}(x) + Y \cdot (k_3 - k_1) + T.\textsf{prefixSum}(bp - 1)$ 
		\ENDIF
		\STATE add the pair $(x', accum - prevaccum-T[bp\bmod w])$ to $\rep{g}$
		\STATE $prevaccum \gets accum$, $x \gets bp$, $k_1 \gets \ceil{\frac{x}{w}}$, $r \gets \max\{ r + 1, \floor{ (1 + \sparapx)r } \}$ 
	\ENDWHILE
	\IF{$(x',y')$ is a start event}
		\STATE $Y \gets Y + y'$
		\STATE $T[x'\bmod w] \gets T[x'\bmod w]+y'$
	\ELSE
		\STATE $Y \gets Y - y'$
		\STATE $T[x'\bmod w] \gets T[x'\bmod w]-y'$
	\ENDIF
	\STATE $accum \gets accum + T.\textsf{suffixSum}(x) + Y \cdot (k_2 - k_1) + T.\textsf{prefixSum}(x' - 1)$
	\STATE $x \gets x'$
\ENDFOR
\end{algorithmic}
\end{algorithm}

\subsection{Generalizing the algorithm of Section~\ref{sec:fptas}}
The only change to the algorithm of Section~\ref{sec:fptas} is that when the set size is small (i.e. when we call the algorithm of Section~\ref{sec:simple}) we use \textsc{InsertAndSparsify} as in the previous subsection.

\paragraph{Time and space complexity.}
We observe that the size of the representation of a $(1 + \apx)$-sum approximation has been changed to $|\Knap_S| = \log_{1 + \sparapx} U^{|S|} = \frac{|S|\log U}{\sparapx}$.
As in Section~\ref{sec:fptas}, we calculate the total time complexity for depth $\log (\sqrt{n\eps})$:
\[\Oh{\sqrt{n\eps} \cdot \rbrackets{\rbrackets{\sqrt{n/\eps}}^3 \rbrackets{\sqrt{\eps/n}}^{-1} \log U \log \rbrackets{\sqrt{n\eps} \rbrackets{\sqrt{\eps/n}}^{-1} \log U}}}\]
which is $\Oh{n^{2.5}\eps^{-1.5}\log U \log (n\eps^{-1}\log U)}$. The total time for depths $i < \log (\sqrt{n\eps})$ is:
\begin{align*} 
	\sum\limits_{i = 0}^{(\log \hminus) - 1} 2^i \rbrackets{\frac{n \log U}{2^{i+1} \cdot \delta_{i+1}}}^2 \log \rbrackets{\frac{n \log U}{2^{i+1} \cdot \delta_{i+1}}}  &\leq \frac{1}{2} \sum\limits_{i = 1}^{\log \hminus} 2^i \rbrackets{\frac{n \log U}{2^i \cdot \di}}^2 \log (n\eps^{-1} \log U) \\
	& = \Oh{\sum\limits_{i = 1}^{\log \hminus} \frac{n^{2.5} \log^2 U}{\eps^{1.5}} \log (n\eps^{-1} \log U) } \\
	& = O(n^{2.5}\varepsilon^{-1.5}\log(n\eps^{-1}\log U)\log (n \eps) \log^2 U).
\end{align*}

The space complexity is dominated by the space used for \textsc{InsertAndSparsify} on a set of size $\sqrt{n/\eps}$ and $\delta = \sqrt{\eps / n}$, that is \Integerspace.

\bibliographystyle{plainurl}

\end{document}

%% file: graphics/apx1.tex
\input{graphics/apx/axes}
\input{graphics/apx/reddots}
\input{graphics/apx/sumf}

\begin{scope}[below]
	\node (x1) at (0,-0.5) { };
	\node (x2) at (2,-0.5) {\phantom{$x_2$}};
	\node (x3) at (2,-1.5) {\phantom{$x_3$}};
	\node (x4) at (6,-0.5) { };
	\node (x5) at (10,-0.5) { };
	\node (x6) at (14,-0.5) {};
\end{scope}

%% file: graphics/apx2.tex
\input{graphics/apx/axes}
\input{graphics/apx/reddots}
\input{graphics/apx/sumf}

\begin{scope}[blue]
	\begin{scope} [left]
		\node (r1) at (0,1) {$r_1$};
		\node (r2) at (0,2) {$r_2$};
		\node (r3) at (0,3) {$r_3$};
		\node (r4) at (0,5) {$r_4$};
		\node (r5) at (0,8) {$r_5$};
		\node (r6) at (0,13) {$r_6$};
	\end{scope}

	\draw[dashed]   (r1) -- (0,1)
					(r2) -- (2,2)
					(r3) -- (2,3)
					(r4) -- (6,5)
					(r5) -- (10,8)
					(r6) -- (14,13);
\end{scope}
\input{graphics/apx/xi}
\input{graphics/apx/bluedots}

%% file: graphics/apx3.tex
\input{graphics/apx/blueaxes}
\input{graphics/apx/reddots}
\input{graphics/apx/sumf}
\input{graphics/apx/bluedots}
\input{graphics/apx/tildes}
\input{graphics/apx/xi}

%% file: graphics/apx4.tex
\input{graphics/apx/blueaxes}
\input{graphics/apx/bludots2}
\input{graphics/apx/tildes}
\input{graphics/apx/xi}